\documentclass[onecolumn]{IEEEtran}
\IEEEoverridecommandlockouts
\usepackage[final]{graphicx}
\usepackage{cite}
\usepackage{amssymb}
\usepackage{amsmath}
\usepackage{amsthm}
\usepackage{amsfonts}
\usepackage{bm}
\usepackage{filecontents}
\usepackage{subfigure}
\usepackage{subfigmat}
\usepackage{makecell}
\usepackage{booktabs}

\begin{document}
\title{Power Minimizer Symbol-Level Precoding:\\ A Closed-Form Sub-Optimal Solution}

\author{Alireza~Haqiqatnejad,~\emph{Student Member,~IEEE},~Farbod~Kayhan,~and Bj\"{o}rn~Ottersten,~\emph{Fellow,~IEEE}
	\thanks{The authors are with Interdisciplinary Centre for Security, Reliability and Trust (SnT), University of Luxembourg, L-1855 Luxembourg. Email:~\{alireza.haqiqatnejad,farbod.kayhan,bjorn.ottersten\}@uni.lu.}\thanks{
    The authors are supported by the Luxembourg National Research Fund (FNR) under CORE Junior project: C16/IS/11332341 Enhanced Signal Space opTImization for satellite comMunication Systems (ESSTIMS).}}

\newtheorem{theorem}{Theorem}
\newtheorem{acknowledgement}[theorem]{Acknowledgement}
\newtheorem{axiom}[theorem]{Axiom}
\newtheorem{case}[theorem]{Case}
\newtheorem{claim}[theorem]{Claim}
\newtheorem{conclusion}[theorem]{Conclusion}
\newtheorem{condition}[theorem]{Condition}
\newtheorem{conjecture}[theorem]{Conjecture}
\newtheorem{corollary}[theorem]{Corollary}
\newtheorem{criterion}[theorem]{Criterion}
\newtheorem{definition}[theorem]{Definition}
\newtheorem{example}[theorem]{Example}
\newtheorem{exercise}[theorem]{Exercise}
\newtheorem{lemma}[theorem]{Lemma}
\newtheorem{notation}[theorem]{Notation}
\newtheorem{problem}[theorem]{Problem}
\newtheorem{proposition}[theorem]{Proposition}
\newtheorem{remark}[theorem]{Remark}
\newtheorem{solution}[theorem]{Solution}
\newtheorem{summary}[theorem]{Summary}
\newtheorem{property}{Property}

\newcommand{\defeq} {\overset{\underset{\mathrm{def}}{}}{=}}
\newcommand{\Deee} {\mathrm{\boldsymbol{\delta}}}
\newcommand{\lamb} {\mathrm{\boldsymbol{\lambda}}}
\newcommand{\ups} {\mathrm{\boldsymbol{\upsilon}}}
\newcommand{\Tht} {\mathrm{\mathbf{\Theta}}}
\newcommand{\g} {\mathrm{\boldsymbol{\gamma}}}
\newcommand{\Gammaaa}{\mathrm{\mathbf{\Gamma}}}
\newcommand{\Del}{\mathrm{\mathbf{\Delta}}}
\newcommand{\Sigmaaa}{\mathrm{\mathbf{\Sigma}}}
\newcommand{\Siii}{\mathrm{\mathbf{\Psi}}}
\newcommand{\Mu} {\mathrm{\boldsymbol{\mu}}}
\newcommand{\omg} {\mathrm{\boldsymbol{\omega}}}
\newcommand{\psiii} {\mathrm{\boldsymbol{\psi}}}
\newcommand{\HHH} {\mathbf{H}}
\newcommand{\bbb}{\mathbf{b}}
\newcommand{\uuu}{\mathbf{u}}
\newcommand{\EEE}{\mathbf{E}}
\newcommand{\WWW}{\mathbf{W}}
\newcommand{\ZZZ}{\mathrm{\mathbf{Z}}}
\newcommand{\I}{\mathrm{\mathbf{I}}}
\newcommand{\A}{\mathbf{A}}
\newcommand{\DDD}{\mathbf{D}}
\newcommand{\FFF}{\mathbf{F}}
\newcommand{\QQQ}{\mathbf{Q}}
\newcommand{\qqq}{\mathbf{q}}
\newcommand{\G}{\mathbf{G}}
\newcommand{\T}{\mathrm{T}}
\newcommand{\F}{\mathrm{F}}
\newcommand{\HH}{\mathrm{H}}
\newcommand{\LLL}{\mathrm{\mathbf{L}}}
\newcommand{\R}{\mathrm{\mathbf{R}}}
\newcommand{\YYY}{\mathrm{\mathbf{Y}}}
\newcommand{\XXX}{\mathrm{\mathbf{X}}}
\newcommand{\OOO}{\mathrm{\mathbf{0}}}
\newcommand{\aaa}{\mathrm{\mathbf{a}}}
\newcommand{\h}{\mathrm{\mathbf{h}}}
\newcommand{\ppp}{\mathrm{\mathbf{p}}}
\newcommand{\s}{\mathrm{\mathbf{s}}}
\newcommand{\vvv}{\mathrm{\mathbf{v}}}
\newcommand{\zzz}{\mathrm{\mathbf{z}}}
\newcommand{\conv}{\mathrm{\mathbf{conv}}}
\newcommand{\dom}{\mathrm{\mathbf{dom}}}
\newcommand{\bd}{\mathrm{\mathbf{bd}}}
\newcommand{\interior}{\mathrm{\mathbf{int}}}
\newcommand{\ccc}{\mathrm{\mathbf{c}}}
\newcommand{\x}{\mathrm{\mathbf{x}}}
\newcommand{\yyy}{\mathrm{\mathbf{y}}}
\newcommand{\EXP}{\mathrm{E}}
\newcommand{\PR}{\mathrm{Pr}}
\newcommand{\TR}{\mathrm{Tr}}
\newcommand{\VEC}{\mathrm{vec}}
\newcommand{\MAX}{\mathrm{max}}
\newcommand{\SGN}{\mathrm{sgn}}
\newcommand{\SUP}{\mathrm{sup}}
\newcommand{\INF}{\mathrm{inf}}
\newcommand{\TT}{\mathrm{T}}
\newcommand{\C}{\mathbb{C}}
\newcommand{\RR}{\mathbb{R}}
\newcommand{\D}{\mathcal{D}}
\newcommand{\SPSK}{\mathrm{-S}\mathrm{PSK}}
\newcommand{\SQPSK}{\mathrm{-S}\mathrm{QPSK}}
\newcommand{\SAPSK}{\mathrm{-S}\mathrm{APSK}}
\newcommand{\SkPSK}{\mathrm{-S}^k\mathrm{PSK}}
\newcommand{\SkQPSK}{\mathrm{-S}^k\mathrm{QPSK}}
\newcommand{\SkAPSK}{\mathrm{-S}^k\mathrm{APSK}}
\newcommand{\PSNR}{\mathrm{PSNR}}
\newcommand{\SNR}{\mathrm{SNR}}
\newcommand{\LLR}{\mathrm{LLR}}
\newcommand{\diag}{\mathop{\mathrm{diag}}}

\maketitle

\begin{abstract}
	
In this letter, we study the optimal solution of the multiuser symbol-level precoding (SLP) for minimization of the total transmit power under given signal-to-interference-plus-noise ratio (SINR) constraints. Adopting the distance preserving constructive interference regions (DPCIR), we first derive a simplified reformulation of the problem. Then, we analyze the structure of the optimal solution using the Karush-Kuhn-Tucker (KKT) optimality conditions, thereby we obtain the necessary and sufficient condition under which the power minimizer SLP is equivalent to the conventional zero-forcing beamforming (ZFBF). This further leads us to a closed-form sub-optimal SLP solution (CF-SLP) for the original problem. Simulation results show that CF-SLP provides significant gains over ZFBF, while performing quite close to the optimal SLP in scenarios with rather small number of users. The results further indicate that the CF-SLP method has a reduction of order $10^3$ in computational time compared to the optimal solution. 
\end{abstract}


\section{Introduction}

Handling multiuser interference (MUI) is the key to increase individual reliable transmission rates of the users in a downlink multiuser channel. Extensive research focusing on multiuser transmit beamforming (precoding) schemes with the aim of suppressing the MUI has been reported in the literature (see e.g. \cite{tb_sol_str} and the references therein). Recently, it has been shown that the MUI may not always be destructive. For example, one may try to exploit the constructive part of the interference, or even converting all the interfering components into constructive interference (CI) \cite{slp_rot}. Subsequently, the conventional viewpoint on multiuser precoding evolved from block-level approaches (based on the channel coherence time) to more meticulous design techniques, such as symbol-dependent precoding which results in significant performance gains. This, however, requires the use of data information (DI) in addition to channel state information (CSI) in order to design the precoder. Such a design concept, which is termed as symbol-level precoding (SLP) \cite{slp_chr, slp_con}, enables the received components associated with the users intended signal to constructively interfere with those of the others.

Despite the performance advantages of SLP schemes, one of the main factors that limits their applicability in practice is the increased complexity at the transmitter side \cite{slp_survey}. Broadly speaking, the SLP module needs to compute the output once per symbol slot, or alternatively, one may design the precoder's output beforehand for every possible combination of users' symbols \cite{slp_con}. In either case, a relatively large number of optimization problems has to be solved for every realization of the time-varying channel. As opposed to conventional schemes, e.g., zero-forcing beamforming (ZFBF), \cite{zfbf_gold}, or minimum mean-square error (MMSE) precoding, \cite{est_kay}, SLP approaches are in general highly computationally demanding. Therefore, reducing the complexity of symbol-level precoders is of great practical importance. This has been addressed in \cite{slp_chr_cf} for the SLP design problem with max-min fairness criterion, where an iterative closed-form method is proposed. In \cite{slp_jev}, the SLP optimization problem minimizing the total transmit power is formulated as a non-negative least squares (NNLS), which can be solved via the existing fast NNLS algorithms. For quadrature amplitude modulation (QAM) schemes, the authors in \cite{slp_per} analyze the structure of the optimal symbol-level precoder with symbol error probability constraints and propose a heuristic low-complexity solution.

Focusing on the distance preserving constructive interference regions (DPCIR) \cite{slp_gen}, in this letter we address the SLP power minimization problem with SINR constraints. By rearranging the original formulation and deriving a simplified version of this problem, we discuss the structure of the optimal solution via the Karush-Kuhn-Tucker (KKT) conditions. This leads us to obtain a closed-form (sub-optimal) expression. We will show that the performance of the closed-form solution is comparable to the optimal SLP for small number of users, but with an extremely reduced complexity. Although the proposed closed-form method is more computationally complex than the ZF scheme, our results show that it provides substantial gains in return. This may indeed stimulate the applicability of the proposed method in realistic scenarios.

The remainder of this letter is organized as follows. In Section \ref{sec:sys}, we describe the downlink multiuser MISO model. We recast the distance preserving CI constraints in Section \ref{sec:ci}. In Section \ref{sec:prob}, we formulate the SLP optimization problem, followed by providing an optimality analysis and proposing a closed-form solution. The simulation results are presented in Section \ref{sec:sim}. Finally, we conclude the paper in Section \ref{sec:con}. 

{\bf{Notations:}} To denote matrices and vectors, we use uppercase and lowercase bold-faced letters respectively, and $[\,\cdot\,]^\T$ denotes the transpose operator. For vectors, $\succeq$ denotes the componentwise inequality. $\mathrm{blkdiag}(\cdot)$ represents a square block matrix having main-diagonal block matrices and zero off-diagonal blocks. The set of non-negative real numbers is represented by $\mathbb{R}_+$.

\section{System Model}\label{sec:sys}

We assume the downlink of a  multiuser unicast channel where a common multiple-antenna transmitter sends independent data streams to $K$ single-antenna users. The transmitter, which is equipped with $N$ antennas, employs a symbol-level precoder for transmission. Accordingly, at a symbol instant, the precoder maps independent data symbols $\{s_k\}_{k=1}^K$ onto $N$ transmit antennas, with $s_k$ denoting the intended symbol for the $k$-th user, and is drawn from a finite equiprobable constellation set. The signal vector to be transmitted is denoted by $\uuu=[u_1,\ldots,u_N]^\T\in\C^{N\times1}$, and is a function of all users' symbols $\{s_k\}_{k=1}^K$. In the presence of frequency-flat fading and additive white Gaussian noise, the received signal at the receiver of the $k$-th user is
\begin{equation}\label{eq:sys}
r_k = \h_k\uuu+z_k, \; k=1,...,K,
\end{equation}
where $\h_k\in\C^{1\times N}$ contains the channel coefficients between the transmit antennas and the single receive antenna of user $k$, and $z_k\sim\mathcal{CN}(0,\sigma_k^2)$ represents the complex Gaussian noise at the $k$-th receiver. The $k$-th user may optimally detect $s_k$ from $r_k$ based on the single-user maximum-likelihood (ML) decision rule. In the following, we adopt the equivalent real-valued notations
$$\tilde{\uuu}=[\Re\{\uuu^\T\},\Im\{\uuu^\T\}]^\T, \; \HHH_k=\begin{bmatrix}
\Re\{\h_k\} \; -\Im\{\h_k\}\\
\Im\{\h_k\} \quad\;\: \Re\{\h_k\}
\end{bmatrix}\!,\vspace{-.1cm}$$
where $\tilde{\uuu}\in\mathbb{R}^{2N\times1}$ and $\HHH_k\in\mathbb{R}^{2\times2N}, k=1,...,K$. Henceforth, the user's noise-free received signal is represented by $\HHH_k \tilde{\uuu}= [\Re\{\h_k\uuu\},\Im\{\h_k\uuu\}]^\T$ for all $k=1,...,K$.

\section{Constructive Interference Constraints}\label{sec:ci}

The DI exploitation in symbol-level design is realized by processing the transmit signal to be received in a specific constructive interference region (CIR) associated to each user. CIRs are generally defined as the regions that preserve or enhance the symbol detection accuracy with respect to (w.r.t.) the original constellation set (see e.g. \cite{slp_chr}). In \cite{slp_gen}, DPCIRs are introduced as a general family of CIRs that do not increase the symbol error rates of the users. The halfspace representation of DPCIRs is provided for generic modulation schemes based on the ML decision regions of the constellation set.

For a constellation point $\x_i$, the corresponding DPCIR is the intersection of a finite number of halfspaces, each associated with one of $M_i$ neighboring points of $\x_i$ (the neighboring points are referred to those points sharing a decision boundary with $\x_i$).  More specifically, from the representation provided in \cite{slp_gen}, it is straightforward to show that any point $\x$ in the DPCIR of $\x_i$ satisfies a linear matrix inequality (LMI) as
\begin{equation}\label{eq:lmi2}
\A_{i} \left(\x - \x_{i}\right) \succeq \OOO,
\end{equation}
where $\A_{i}\in\RR^{M_i\times 2}$ is a matrix that contains the normal vectors of DPCIR boundaries (hyperplanes), given by
\begin{equation}\label{eq:A}
\A_{i} = \begin{bmatrix} \aaa_{i,1}^\T \\ \vdots \\ \aaa_{i,M_{i}}^\T \end{bmatrix} = \begin{bmatrix} (\x_{i} - \x_{i,1})^\T \\ \vdots \\ (\x_{i} - \x_{i,M_{i}})^\T \end{bmatrix},
\end{equation}
with $\x_{i,1},...,\x_{i,M_{i}}$ denoting the neighboring constellation points of $\x_{i}$. By introducing a non-negative vector $\Deee\in\RR_+^{M_i\times 1}$, \eqref{eq:lmi2} is equivalent to
\begin{equation}\label{eq:lme2}
\A_{i} \left(\x - \x_{i}\right) = \Deee_i, \; \Deee_i\succeq\OOO,
\end{equation}
which will be used as the CI constraint in our formulation of the SLP optimization problem. It is also shown that if $\x_{i}$ is a constellation point with bounded decision region, we always have $\Deee_i=\OOO$ \cite{slp_gen}. On the other hand, if $\x_{i}$ refers to a constellation point with unbounded decision region, it is shown in \cite{slp_tsp} that the associated DPCIR is a polyhedral angle with two infinite edges starting from $\x_{i}$. Hence, there always exist exactly two halfspaces that are sufficient to characterize the DPCIR, and to construct $\A_{i}\in\mathbb{R}^{2\times2}$ as a non-singular matrix (see \cite{slp_tsp} for more details). Further, any point $\x$ that satisfies \eqref{eq:lme2} can be uniquely specified by $\Deee_i\in\mathbb{R}_+^{2\times1}$.

\section{SINR-Constrained power minimizer SLP}\label{sec:prob}

In this section, we are interested in the symbol-level power minimization problem being constrained by CIRs as well as individual SINR requirements, i.e.,
\begin{equation}\label{eq:raw}
\begin{aligned}
\underset{\uuu}{\mathrm{minimize}} & \quad \uuu^\T\uuu \\
\mathrm{s.t.} &\quad \h_k\uuu\in \sigma_k\sqrt{\gamma_k}\;\D_k,\;k = 1,...,K,
\end{aligned}
\end{equation}
where $\gamma_k$ and $\D_k$ respectively denote the required SINR and the CIR related to the $k$-th user. By assuming DPCIRs, the CI expressions in the form of \eqref{eq:lme2} can be used to explicitly define the CI constraints of \eqref{eq:raw}.

In order to simplify the analysis and notation, and without loss of generality, we assume an identical $M$-ary constellation set $\mathcal{X}=\{\x_i|\x_i\in\mathbb{R}^{2\times1}\}_{i=1}^M$ with unit average power for all $K$ users. A symbol $s_k$ then corresponds to one of the points $\{\x_i\}_{i=1}^M$ in $\mathcal{X}$. We denote by $i_k$ the index of the constellation point corresponding to $s_k$, i.e.,\vspace{-.1cm}
$$\x_{i_k}=[\Re\{s_k\},\Im\{s_k\}]^\T, \; i_k\in\{1,...,M\}, \; k=1,...,K.\vspace{-.1cm}$$
In the rest, we consider the case in which all $\x_{i_k}$ have unbounded decision regions, i.e., $\Deee_{i_k}\in\mathbb{R}_+^{2\times1},k=1,...,K$. It is straightforward to generalize the results to the case with the users' symbols comprising a mix of points with both bounded and unbounded decision regions.

For each user $k$, $\HHH_k \tilde{\uuu}$ is pushed by the precoder to be inside the corresponding DPCIR up to a scale factor that depends on the given SINR requirement. 
From \eqref{eq:lme2}, by substituting $\HHH_k \tilde{\uuu}$ for $\x$ and replacing the scaled constellation point $\sigma_k\sqrt{\gamma_k}\,\x_{i_k}$, we obtain the CI constraint for the $k$-th user as
\begin{equation}\label{eq:lme}
\A_{i_k} (\HHH_k \tilde{\uuu} - \sigma_k\sqrt{\gamma_k}\,\x_{i_k}) = \Deee_{i_k}, \; \Deee_{i_k}\succeq\OOO.
\end{equation}
Taking all the users into account, by stacking the CI constraint \eqref{eq:lme} for all $k\in\{1,...,K\}$ into the matrix form, we have
\begin{equation}\label{eq:lmeall}
\A (\tilde\HHH \tilde{\uuu} - \Sigmaaa\Gammaaa^{1/2}\,\tilde{\x}) = \Deee, \; \Deee\succeq\OOO,
\end{equation}
where $\tilde{\HHH} \!\triangleq\! [\HHH_1^\T,...,\HHH_K^\T]^\T, \A \!\triangleq\! \mathrm{blkdiag}(\A_{i_1},...,\A_{i_K})$, $\Sigmaaa \triangleq \mathrm{blkdiag}(\sigma_1,...,\sigma_K)$, $\Gammaaa \triangleq \mathrm{blkdiag}(\gamma_1,...,\gamma_K)$, $\tilde{\x}\triangleq[\x_{i_1},...,\x_{i_K}]^\T$, and $\Deee\triangleq[\Deee_{i_1},...,\Deee_{i_K}]^\T$. Recalling that all $\A_{i_k},k=1,...,K$ are non-singular, it can be verified that $\A$ is invertible. Therefore, \eqref{eq:lmeall} can be written as
\begin{equation}\label{eq:lmeall2}
\tilde{\HHH} \tilde{\uuu} = \Sigmaaa\Gammaaa^{1/2}\,\tilde{\x} + \A^{-1}\Deee, \; \Deee\succeq\OOO.
\end{equation}
Using \eqref{eq:lmeall2}, a (convex) compact formulation for the optimization problem of interest can be expressed as
\begin{equation}\label{eq:pm3}
\begin{aligned}
\underset{\tilde{\uuu},\Deee\succeq\OOO}{\mathrm{minimize}} & \quad \tilde{\uuu}^\T\tilde{\uuu}\\
\mathrm{s.t.} & \quad \tilde{\HHH} \tilde{\uuu} = \Sigmaaa\Gammaaa^{1/2} \; \tilde{\x} + \A^{-1}\Deee.
\end{aligned}
\end{equation}
We further notice that $\tilde{\HHH}$ is (almost surely) a full row rank matrix due to its stochastic nature. This results in a bijection between $\tilde{\uuu}$ and $\Deee$ in \eqref{eq:lmeall2}, i.e., for any given $\Deee$, the (least-norm) $\tilde{\uuu}$ is obtained by
\begin{equation}\label{eq:u}
\mathrm{arg}\underset{\tilde{\uuu}}{\mathrm{min}}\;\tilde{\uuu}^\T\tilde{\uuu} \triangleq \tilde{\uuu}_\text{\tiny LN}(\Deee) = \underbrace{\tilde{\HHH}^\dagger\Sigmaaa\Gammaaa^{1/2} \; \tilde{\x}}_{\tilde{\uuu}_\text{\tiny ZF}} + \underbrace{\tilde{\HHH}^\dagger\A^{-1}\Deee}_{\tilde{\uuu}_\text{\tiny SL}},
\vspace{-.15cm}
\end{equation}
where $\tilde{\HHH}^\dagger = \tilde{\HHH}^\T(\tilde{\HHH}\tilde{\HHH}^\T)^{-1}$ is the Moore-Penrose inverse of $\tilde{\HHH}$. Equation \eqref{eq:u} gives the structure of the power minimizer transmit vector, i.e., the optimal solution of \eqref{eq:pm3}. Intuitively, it consists of two parts: $\tilde{\uuu}_\text{\tiny ZF}$, which is the solution of constructive interference zero-forcing (CIZF) \cite{slp_con}, and $\tilde{\uuu}_\text{\tiny SL}$, the DPCIR-dependent part. It is worth noting that under equal SINR thresholds, the CIZF precoder refers to the conventional ZFBF.
Accordingly, the optimization problem \eqref{eq:pm3} reduces to
\vspace{-.1cm}
\begin{equation}\label{eq:pm4}
\Deee^*=\mathrm{arg}\underset{\Deee\succeq\OOO}{\mathrm{min}} \quad \|\tilde{\HHH}^\dagger\Sigmaaa\Gammaaa^{1/2} \; \tilde{\x} + \tilde{\HHH}^\dagger \A^{-1}\Deee\|^2.
\end{equation}
The optimal transmit vector is then obtained by
\begin{equation}\label{eq:uopt}
\tilde{\uuu}^* = \tilde{\uuu}_\text{\tiny LN}(\Deee^*).
\end{equation}
Problem \eqref{eq:pm4} is classified as a quadratic programming (QP) \cite{convex_boyd}.
There are various algorithms to solve a QP optimization, e.g., interior-point or active set methods \cite{qp_methods}; however, they typically require a rather large number of iterations to converge. This motivates us to find a more tractable, possibly sub-optimal, solution for \eqref{eq:pm4}.

{\emph{Optimality Analysis}}: The Lagrangian of \eqref{eq:pm4} is given by
\begin{align}\label{eq:L}
\mathcal{L}(\Deee,\lamb) &\!=\! \tilde{\x}^\T\Sigmaaa\Gammaaa^{1/2}\tilde{\HHH}^{\dagger\T}\tilde{\HHH}\Sigmaaa\Gammaaa^{1/2}\tilde{\x} \nonumber\\
&\!+\! \Deee^\T\A^{-\T}\tilde{\HHH}^{\dagger\T}\tilde{\HHH}^\dagger\Sigmaaa\Gammaaa^{1/2}\tilde{\x} + \tilde{\x}^\T\Sigmaaa\Gammaaa^{1/2}\tilde{\HHH}^{\dagger\T}\tilde{\HHH}^\dagger\A^{-1}\Deee \nonumber\\
&\!+\! \Deee^\T\A^{-\T}\tilde{\HHH}^{\dagger\T}\tilde{\HHH}^\dagger\A^{-1}\Deee + \lamb^\T\Deee,
\vspace{-.1cm}
\end{align}
from which the Lagrange dual problem can be written as
\begin{equation}\label{eq:pmd}
\underset{\lamb\preceq\OOO}{\mathrm{max}} \quad \underset{\Deee\succeq\OOO}{\INF} \; \mathcal{L}(\Deee,\lamb),
\end{equation}
where $\lamb$ is the vector of the Lagrange multipliers. Denoting the primal and dual optimals by $\Deee^*$ and $\lamb^*$, respectively, the Karush-Kuhn-Tucker (KKT) optimality conditions are
\begin{subequations}\label{eq:kkt}
\begin{align}
\nabla_\Deee\,\mathcal{L}(\Deee^*,\lamb^*) &= \OOO, \quad (\mathrm{stationarity}) \label{eq:kkt1}\\
\Deee^*&\succeq\OOO, \quad (\mathrm{primal \; feasibility}) \label{eq:kkt2}\\
\lamb^*&\preceq\OOO, \quad (\mathrm{dual \; feasibility}) \label{eq:kkt3}\\
\lamb^{*T}\Deee^*&=0, \quad (\mathrm{complementary \; slackness}) \label{eq:kkt4}
\end{align}
\end{subequations}
Notice that since the primal problem \eqref{eq:pm4} is convex, strong duality holds and the KKT conditions \eqref{eq:kkt1}-\eqref{eq:kkt4} are necessary and sufficient \cite{convex_boyd}. Consequently, any candidate solution satisfying all the KKT conditions is globally optimal.

Letting $\QQQ = \QQQ^{\T}\triangleq\A^{-\T}\tilde{\HHH}^{\dagger\T}\tilde{\HHH}^\dagger\A^{-1}=[\qqq_1,...,\qqq_{2K}]^\T$ and $\vvv\triangleq\A^{-\T}\tilde{\HHH}^{\dagger\T}\tilde{\HHH}^\dagger\Sigmaaa\Gammaaa^{1/2}\tilde{\x}=[v_1,...,v_{2K}]^\T$, the stationarity condition \eqref{eq:kkt1} is derived as $2\QQQ\Deee^* + 2\vvv + \lamb^* = \OOO$ and thus,
\begin{equation}\label{eq:sta2}
\lamb^* = -2(\QQQ\Deee^* + \vvv).
\end{equation}
It then follows from \eqref{eq:kkt3} and \eqref{eq:sta2} that
\begin{equation}\label{eq:sta3}
\QQQ\Deee^* + \vvv \succeq \OOO.
\end{equation}
Furthermore, substituting $\lamb^*$ from \eqref{eq:sta2} in \eqref{eq:kkt4} yields
\begin{equation}\label{eq:sta4}
(\QQQ\Deee^* + \vvv)^\T \Deee^* = 0,
\end{equation}
from which by denoting $\psiii\!\triangleq\!\QQQ\Deee^* + \vvv\!=\![\psi_1,...,\psi_{2K}]^\T$ and $\Deee^* = [\delta^*_1,...,\delta^*_{2K}]^\T$, it follows that
\begin{equation}\label{eq:sta5}
\sum_{l=1}^{2K}\psi_l \, \delta^*_l = 0.
\end{equation}
Considering \eqref{eq:kkt2} and \eqref{eq:sta3}, we have $\psi_l\geq0,l=1,...,2K$. As a consequence, the optimality condition \eqref{eq:sta5} is met iff
\begin{equation}\label{eq:sta6}
\psi_l \, \delta^*_l=0, \; \forall l\in\{1,...,2K\}.
\end{equation}
In other words, $\psi_l$ and $\delta^*_l$ cannot be both non-zero for any specific $l\in\{1,...,2K\}$. Based on this observation, the following lemma relates the SLP solution to that of ZFBF.
\begin{lemma}\label{lem:1}
	The optimal solution of the SLP power minimization \eqref{eq:pm3} is identical to the solution of ZFBF iff $\vvv\succeq\OOO$.
\end{lemma}
\begin{proof}
	{\emph{Sufficiency}}: It is clear from \eqref{eq:u} and \eqref{eq:uopt} that $\tilde{\uuu}^*$ equals the solution of ZFBF iff $\Deee^* = \OOO$. Given $\vvv\succeq\OOO$, let assume by contradiction that $\Deee^* \neq \OOO$, i.e., $\exists \, l$ such that $\delta_l^* > 0$, which gives $\vvv^\T \Deee^* \geq 0$. Let us rewrite the optimality condition \eqref{eq:sta4} as
	$\Deee^{*\T}\QQQ\Deee^* + \vvv^\T \Deee^* = 0$. By definition, $\QQQ$ is symmetric and $\QQQ = (\tilde{\HHH}^\dagger\A^{-1})^\T\tilde{\HHH}^\dagger\A^{-1}$, where $\tilde{\HHH}^\dagger\A^{-1}$ has full column rank due to the random concatenated channel $\tilde{\HHH}$. Hence, $\QQQ$ is a positive definite matrix \cite[Theorem 7.2.7]{horn_matrix}, i.e., $\Deee^{*\T}\QQQ\Deee^* > 0$ for any $\Deee^* \neq \OOO$. This, however, yields $\Deee^{*\T}\QQQ\Deee^* + \vvv^\T \Deee^* > 0$ which contradicts the KKT condition \eqref{eq:sta4}. Therefore, having $\vvv\succeq\OOO$, it necessarily holds that $\Deee^* = \OOO$, as required.
	
	{\emph{Necessity}}: Assuming $\Deee^* = \OOO$, it immediately follows from \eqref{eq:sta3} that $\vvv\succeq\OOO$. This completes the proof.
\end{proof}
Lemma \ref{lem:1} provides the necessary and sufficient condition under which the DPCIR-based SLP has the same solution as ZFBF. This occurs depending on the instantaneous realization of the users' symbols as $\vvv\succeq\OOO$ is equivalently met by $\A^{-\T}\tilde{\HHH}^{\dagger\T}\tilde{\HHH}^\dagger\Sigmaaa\Gammaaa^{1/2}\tilde{\x}\succeq\OOO$. It can be further inferred from \eqref{eq:sta6} and Lemma \ref{lem:1} that as the number of non-zero (i.e. positive) elements of $\psiii$ decreases, the SLP solution may diverge from that of the ZFBF. In the extreme case with $\psi_l=0,l=1,...,2K$, there exists at least one $\delta^*_l \neq 0$. It can be verified as $\psiii = \QQQ\Deee^* + \vvv = \OOO$ has a unique solution equal to $\Deee^* = -\QQQ^{-1}\vvv = \A\Sigmaaa\Gammaaa^{1/2}\tilde{\x}$. Since $\A$ is full rank, it has an empty null space, thus $\A\Sigmaaa\Gammaaa^{1/2}\tilde{\x} \neq \OOO$. This means that $\Deee^* \neq \OOO$ and it has at least one non-zero entry. In such case, the SLP design results in higher precoding gains compared to the ZFBF. This case, however, is feasible only if the unique solution to the system of linear equations $\QQQ\Deee^* + \vvv = \OOO$ is non-negative, i.e., $-\QQQ^{-1}\vvv\succeq\OOO$, or equivalently $\A\Sigmaaa\Gammaaa^{1/2}\tilde{\x}\preceq\OOO$.

{\emph{Closed-Form Sub-Optimal Solution}}: By utilizing the KKT optimality analysis, a low-complexity solution can be derived with a simple idea behind. Let $\mathcal{Z}\triangleq\{l|\delta^*_l\neq0\}$ be the set of inactive constraints at the optimum, then according to \eqref{eq:sta6}, we have
\begin{equation}\label{eq:z}
\psi_l = \qqq_l^\T\Deee^*+v_l = 0, \; \forall \, l\in\mathcal{Z}.
\vspace{-.1cm}
\end{equation}
which gives a reduced system of linear equations to obtain $\Deee^*$. However, in theory, we do not have such prior information, i.e., the inactive set $\mathcal{Z}$ is not explicitly known. Instead, we can use a presumption on $\mathcal{Z}$ as follows.

From \eqref{eq:sta4}, the positive definiteness of $\QQQ$ verifies that
\begin{equation}\label{eq:app}
\vvv^\T \Deee^* =\sum_{l=1}^{2K}v_l \, \delta^*_l \leq 0,
\vspace{-.1cm}
\end{equation}
where equality holds only for $\Deee^*=\OOO$. A presumption of the inactive set $\mathcal{Z}$ can be derived based on the sign of the elements in $\vvv$, i.e., $\hat{\mathcal{Z}}=\{l|v_l<0\}$ with $|\hat{\mathcal{Z}}|=L$. Here, it is assumed that $\delta_l^*=0$ (i.e., the $l$-th constraint is active) for those $l$ with $v_l\geq0$. This results in a reduced system of linear equations as
\begin{equation}\label{eq:red}
\QQQ'\Deee^{*\prime} + \vvv' = \OOO,
\end{equation}
where $\QQQ'\in\mathbb{R}^{L\times L}$, $\Deee^{*\prime}\in\mathbb{R}^{L\times 1}$ and $\vvv'\in\mathbb{R}^{L\times 1}$ are punctured versions of $\QQQ$, $\Deee^*$ and $\vvv$ obtained by excluding the equations, variables and coefficients corresponding to $l\notin\hat{\mathcal{Z}}$. This new system  has $L$ linear equations and $L$ variables, where $L \leq 2K$, hence possibly less dimensionality than the original problem. By noticing that the non-singularity of $\QQQ$ is preserved under puncturing, the (unique) solution of \eqref{eq:red} is readily given by the following closed-form expression
\begin{equation}\label{eq:cl}
\Deee^{*\prime} = \MAX\{-\QQQ^{\prime -1} \vvv',\OOO\},
\end{equation}
where $\MAX\{\cdot\}$ denotes the elementwise maximum, and is applied in order to ensure the satisfaction of the primal feasibility condition \eqref{eq:kkt2}. The entire $\Deee^*$ can be obtained by inserting the zero entries  $\delta^*_l,l\notin\hat{\mathcal{Z}}$ into $\Deee^{*\prime}$. The resulting $\Deee^*$ is then a sub-optimal solution for \eqref{eq:pm4}.

As for generalization of the proposed method, it should be noted that for generic constellations possibly having points with bounded decision regions, by introducing a diagonal binary weighting matrix $\WWW$ for $\Deee$ in the optimization problem \eqref{eq:pm3}, similar analyses hold.

\section{Simulation Results}\label{sec:sim}

In this section, we evaluate the performance of the proposed closed-form solution (CF-SLP). The results are compared with the optimal power minimizer SLP (OPT-SLP) and the CIZF scheme. We assume equal SINR thresholds for all the users, i.e., $\gamma_k=\gamma,k=1,...,K$, hence CIZF and ZFBF schemes are equivalent. We further assume scenarios with $N=K$, and it is referred to as system dimension. The complex channel vectors are randomly generated as $\h_k\sim\mathcal{CN}(\OOO,\I),k=1,...,K$. All the simulations are done using MATLAB software and CVX convex optimization package \cite{cvx} (SDPT3 solver).

In Fig. \ref{fig:power} (a)-(c), we plot the total transmit power as a function of SINR threshold for three system dimensions, employing the QPSK modulation scheme. As it can be observed, for $2\times2$ and $4\times4$ scenarios, CF-SLP is almost equally as power-efficient as OPT-SLP. The loss due to the sub-optimality of CF-SLP is around $0.1$ dBW and $0.4$ dBW in the depicted range of SINR for $2\times2$ and $4\times4$ system dimensions, respectively. This loss is shown in Fig. \ref{fig:power} (c) to be larger for the $8\times8$ scenario (with a maximum loss of $4.4$ dBW), which can be explained by the possibility of having more errors in $\hat{\mathcal{Z}}$ (w.r.t. $\mathcal{Z}$) as the dimensionality of problem increases. Nevertheless, CF-SLP still offers performance gains up to $3.2$ dBW compared to CIZF (ZFBF).

\begin{figure}
	\begin{subfigmatrix}{4}
		\subfigure[]{\includegraphics[width=.49\columnwidth]{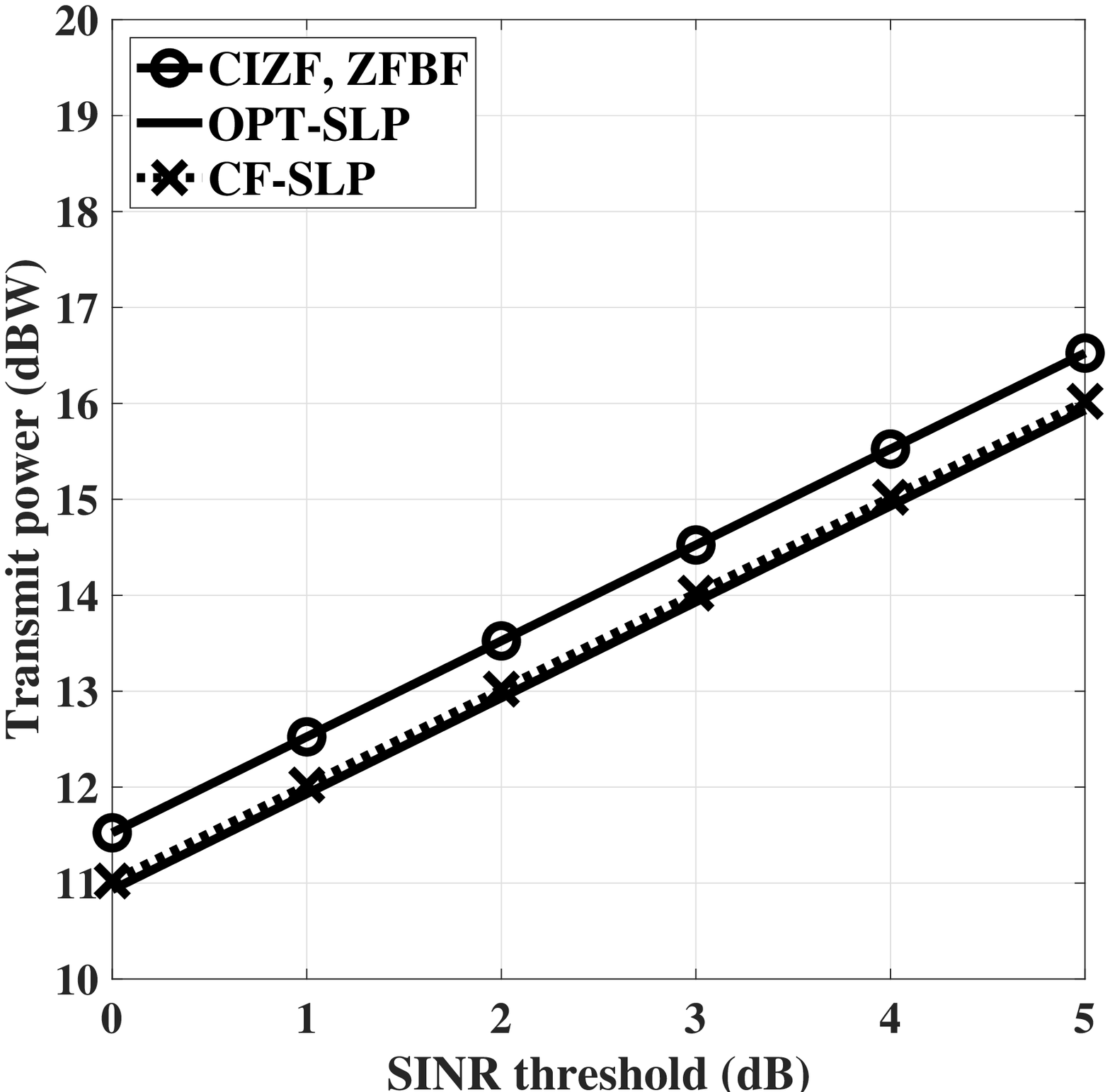}}
		\subfigure[]{\includegraphics[width=.49\columnwidth]{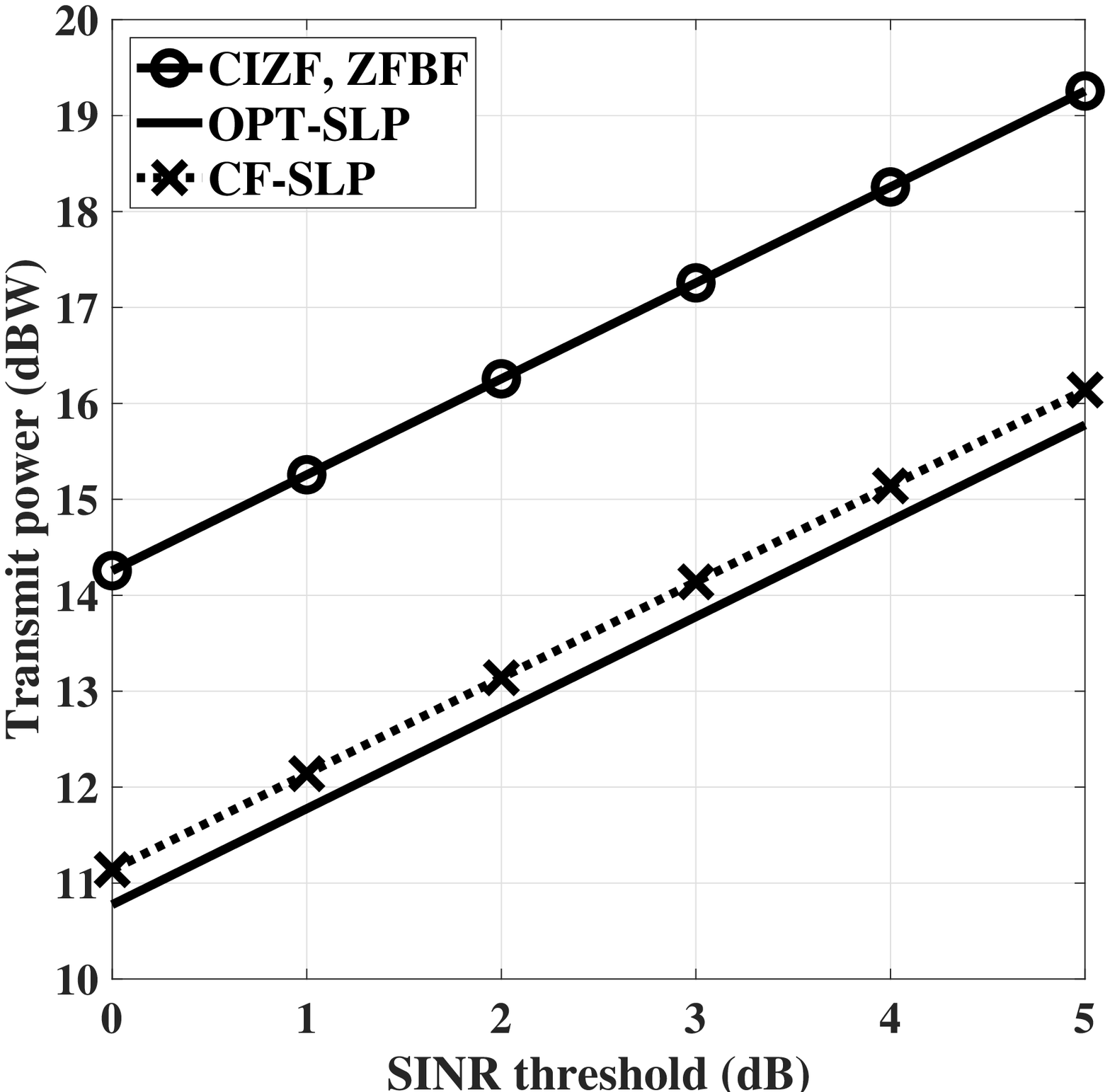}}
		\subfigure[]{\includegraphics[width=.49\columnwidth]{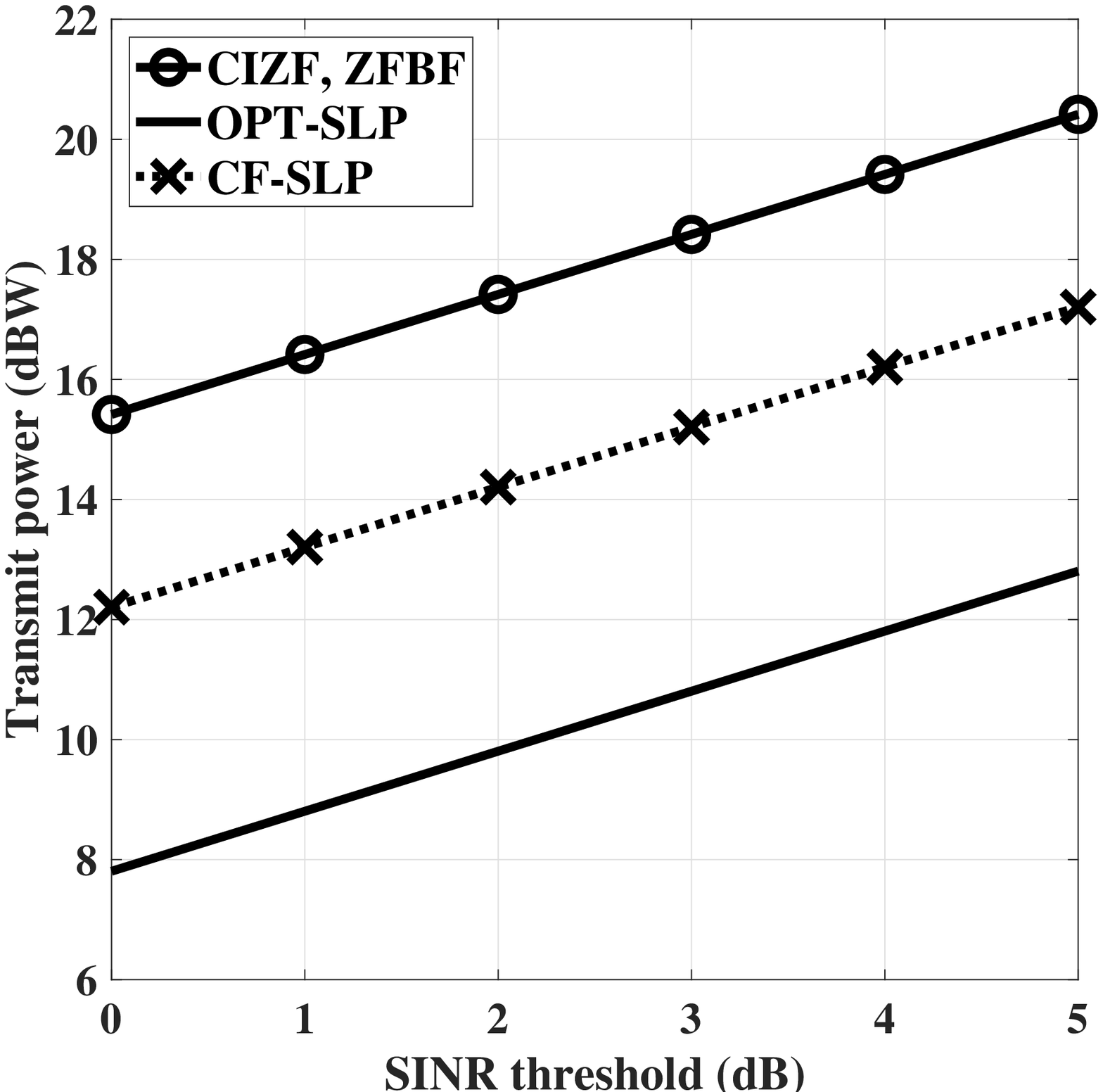}}
		\subfigure[]{\includegraphics[width=.49\columnwidth]{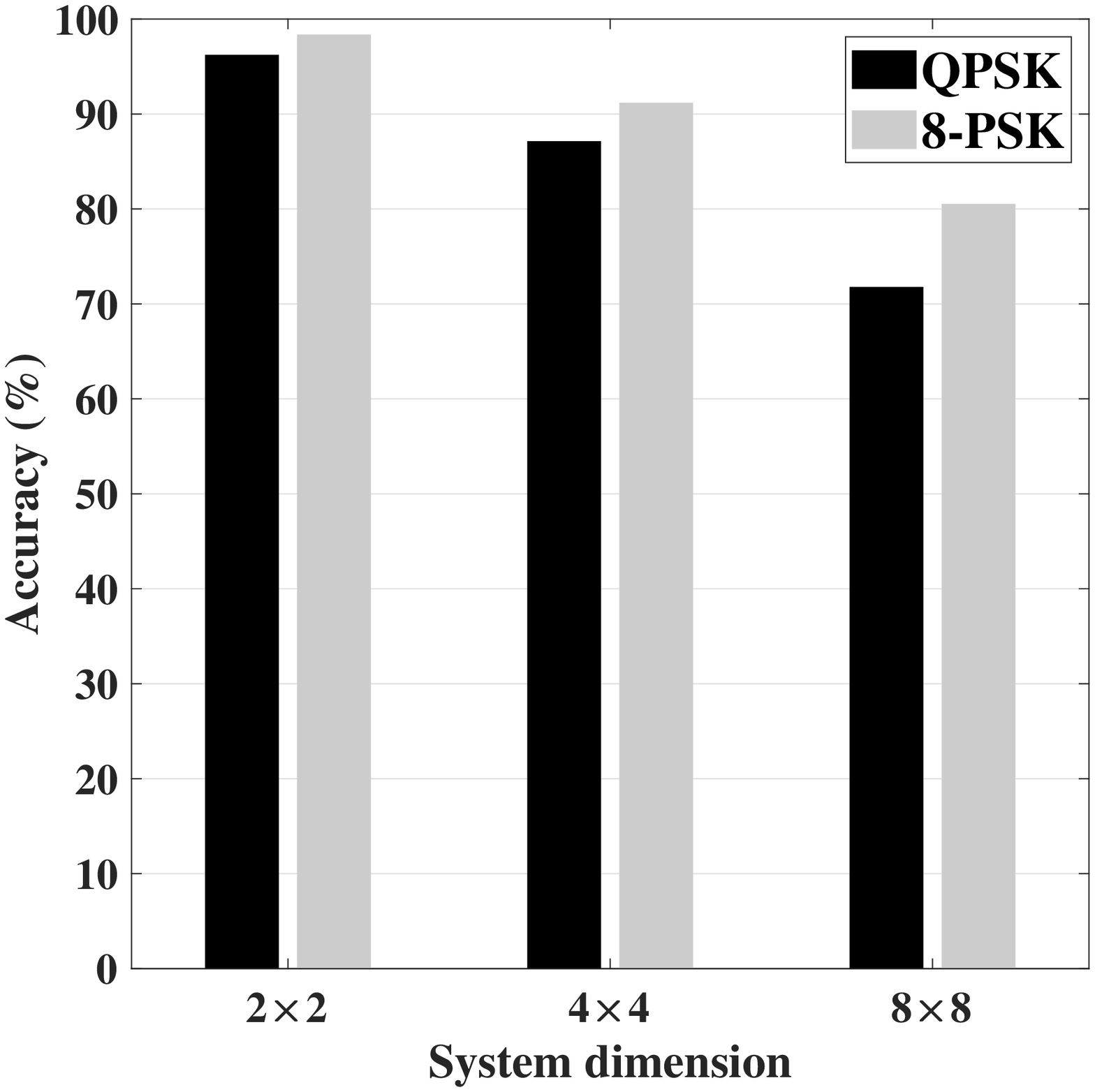}}
	\end{subfigmatrix}
	\caption{Total transmit power vs. SINR threshold (a) $2\times2$ (b) $4\times4$ (c) $8\times8$. (d) Accuracy of CF-SLP for different system dimensions with $\gamma = 3$ dB.}
	\label{fig:power}
\end{figure}

The accuracy percentage of CF-SLP is plotted in Fig. \ref{fig:power} (d), where the measure of accuracy is the number of correct entries in $\hat{\mathcal{Z}}$ as compared to the optimum. It can be seen that for $2\times2$ scenario, CF-SLP performs with an accuracy of at least $95\%$. This accuracy drops down to $70-80\%$ for an $8\times8$ system, where the results show greater accuracies for higher orders of PSK constellation. 
However, notice that CF-SLP, as obtained in \eqref{eq:cl}, guarantees the DPCIR constraints to be met; therefore it does not degrade the symbol error rate performance.

In Table \ref{tab:time}, we compare the complexity of the three precoding schemes in terms of the average execution time per symbol slot (the time values are obtained via the relevant functions of MATLAB and CVX). As for the ZFBF scheme, the precoding matrix is multiplied by the users' symbol vector every symbol time (the precoding matrix computation, which is typically done depending on the CSI validity interval, is not included in the execution times). The CF-SLP method consists of computing and puncturing $\QQQ$ and $\vvv$, and then solving \eqref{eq:cl}. On the other hand, solving the convex power optimization \eqref{eq:pm4} accounts for the execution time of OPT-SLP. The numerical results show that CF-SLP can potentially reduce the complexity of the precoding module by orders of $10^3$, compared to OPT-SLP. In comparison with ZFBF, the execution times of CF-SLP are expectedly greater (by orders of $10$) due to the symbol-level processing. This indicates a performance-complexity tradeoff between the two schemes, especially for large system dimensions.

\begin{table}[h]
	\caption{Execution time of the precoding schemes.}
	\label{tab:time}
	\centering
	\renewcommand{\arraystretch}{0.7}
	\begin{tabular}{ccccc}
		\toprule
		Modulation & Dimension & \multicolumn{3}{c}{Execution time (ms/symbol)}\\
		\cmidrule{3-5}
		& & OPT-SLP & CF-SLP & ZFBF\\
		\midrule
		QPSK & $2\times2$ & 642.0 & 0.086 & 0.004\\
		& $4\times4$ & 678.5 & 0.104 & 0.005\\
		& $8\times8$ & 692.2 & 0.258 & 0.007\\
		\midrule
		8-PSK & $2\times2$ & 751.7 & 0.090 & 0.005\\
		& $4\times4$ & 713.6 & 0.101 & 0.007\\
		& $8\times8$ & 715.8 & 0.271 & 0.009\\
		\bottomrule
	\end{tabular}
\end{table}

\section{Conclusion}\label{sec:con}

In this letter, we propose a closed-form sub-optimal solution method for the power minimizer SLP with SINR constraints. This is done by first simplifying the original formulation, and then discussing the optimality of the reduced problem via the KKT conditions. Through the optimality analysis, we obtain the necessary and sufficient condition that the power minimizer SLP has the same solution as ZFBF. The analysis further helps us to derive a closed-form sub-optimal SLP solution (CF-SLP). In particular, for scenarios with rather small number of users, the results show that CF-SLP offers quite similar performance compared to the computationally demanding optimal SLP, and therefore it can be an appropriate alternative for the conventional ZFBF scheme.



\end{document}